\documentclass[11pt]{article}
\usepackage{amsmath, amssymb, amsthm}
\usepackage{graphicx}
\usepackage{enumerate}
\usepackage{enumitem}
\usepackage{xcolor}
\setlist{nosep}
\usepackage{natbib}
\usepackage{url} 
\usepackage{multirow}

\DeclareMathAlphabet{\bbold}{U}{bbold}{m}{n}
\newcommand{\1}{\ensuremath{\bbold{1}}}

\usepackage{caption}
\usepackage{setspace}

\DeclareCaptionFont{single}{\setstretch{1}}
\newcommand{\blind}{1}
\newcommand{\argmin}{\operatornamewithlimits{argmin}}
\newtheorem{theorem}{Theorem}

\begin{document}

\def\spacingset#1{\renewcommand{\baselinestretch}%
{#1}\small\normalsize} \spacingset{1}


\if1\blind
{
  \title{\bf Calibrated Bayesian Nonparametric Tolerance Intervals}
  \author{Tony Pourmohamad\\
    Data and Statistical Sciences, AbbVie\\
    and \\
    Robert Richardson \\
    Department of Statistics, Brigham Young University\\
    and \\
    Bruno Sans\'{o} \\
    Department of Statistics, University of California, Santa Cruz}
  \maketitle
} \fi

\if0\blind
{
  \bigskip
  \bigskip
  \bigskip
  \begin{center}
    {\LARGE\bf Calibrated Bayesian Nonparametric Tolerance Intervals}
\end{center}
  \medskip
} \fi

\bigskip
\begin{abstract}
Tolerance intervals provide bounds that contain a specified proportion of a population with a given confidence level, yet their construction remains challenging when parametric assumptions fail or sample sizes are small. Traditional nonparametric methods, such as Wilks’ intervals, lack flexibility and often require large samples to be valid. We propose a fully nonparametric approach for constructing one-sided and two-sided tolerance intervals using a calibrated Gibbs posterior. Leveraging the connection between tolerance limits and population quantiles, we employ a Gibbs posterior based on the asymmetric Laplace (check) loss function. A key feature of our method is the calibration of the learning rate, which ensures nominal frequentist coverage across diverse distributional shapes. Simulation studies show that the proposed approach often yields shorter intervals than classical nonparametric benchmarks while maintaining reliable coverage. The framework’s practical utility is illustrated through applications in ecology, biopharmaceutical manufacturing, and environmental monitoring, demonstrating its flexibility and robustness across diverse applications.
\end{abstract}

\noindent%
\vfill

\newpage
\spacingset{1.9} 

\setlength{\abovedisplayskip}{3pt}
\setlength{\belowdisplayskip}{3pt}
\setlength{\abovedisplayshortskip}{0pt}
\setlength{\belowdisplayshortskip}{3pt}

\section{Introduction}
\label{sec:intro}
Tolerance intervals (TIs) play a fundamental role in statistical practice by providing bounds that contain a specified proportion of a population with a prescribed level of confidence \citep{guttman1970statistical, meeker2017statistical}. Unlike confidence intervals, which quantify uncertainty about unknown parameters, tolerance intervals are designed to control population coverage and are therefore particularly relevant in quality control \citep{hamada:2002, wangtsung:2017, yao:2020}, pharmaceutical manufacturing \citep{dong:2015, dong:2015b, schwenke:2021, oliva:2024}, and engineering \citep{sermer:2014, hofer:2017, chen:2020}. Despite their long history, the construction of tolerance intervals remains challenging outside parametric settings, especially when distributional assumptions are difficult to justify.

Classical parametric tolerance intervals rely on strong modeling assumptions and can exhibit substantial sensitivity to misspecification. Nonparametric alternatives, such as Wilks-type tolerance intervals \citep{wilks:1941}, avoid distributional assumptions but typically require large sample sizes and offer limited flexibility. In particular, existing nonparametric methods are often restricted to fixed forms of one-sided or two-sided intervals and do not readily accommodate alternative definitions of coverage, such as those targeting specific population quantiles rather than aggregate mass. These limitations motivate the development of principled nonparametric approaches that provide both flexibility and reliable frequentist guarantees.

A key insight underlying tolerance interval construction is its intimate connection with population quantiles \citep{patel:1986}. In the univariate setting, a one-sided tolerance bound corresponds directly to inference on a population quantile, while two-sided tolerance intervals can be represented in terms of pairs of quantiles. From this perspective, the construction of tolerance intervals reduces to the problem of performing inference on quantile functionals of an unknown distribution. This viewpoint suggests that advances in quantile inference may be leveraged to develop improved tolerance interval procedures \citep{taddy:2010}.

In this paper, we propose a nonparametric generalized Bayesian approach \citep{bissiri:2016} for constructing tolerance intervals based on Gibbs posterior inference for population quantiles \citep{martin:2022}. Our method uses the check (pinball) loss function to target quantiles directly \citep{koenker:1978}, without requiring a likelihood or parametric model for the data-generating process. The resulting Gibbs posterior provides a coherent representation of uncertainty for quantile functionals, which can be translated into uncertainty about tolerance bounds. In this way, the framework yields a posterior-based procedure while remaining fully nonparametric.

We further demonstrate that explicitly calibrating the learning rate $\eta$ of the Gibbs posterior ensures the intervals achieve nominal frequentist coverage while producing intervals that are, on average, consistently shorter than those of traditional nonparametric benchmarks. Moreover, this framework naturally extends to quantile-defined coverage, a task for which standard order-statistic-based methods are ill-suited. Through simulation studies and real-world applications in ecology, biopharmaceuticals, and environmental health, we illustrate the robustness of the proposed approach across diverse distributional shapes and small-sample settings.

The remainder of the paper is organized as follows. Section 2 reviews the foundational concepts of tolerance intervals and the Gibbs posterior framework. Section 3 details our methodology for constructing one-sided and two-sided tolerance intervals, introducing a joint posterior construction for the two-sided case. In Section 4, we describe the calibration algorithm for the learning rate $\eta$, which ensures the intervals achieve the target frequentist coverage. Section 5 presents simulation studies comparing our method to classical parametric and nonparametric benchmarks across various distributions. Section 6 demonstrates the practical utility of the calibrated Gibbs posterior through three real-world applications. Finally, Section 7 concludes with remarks and directions for future research.

\section{Tolerance Intervals and Population Quantiles}\label{TI}
Let $Y$ denote a univariate random variable with unknown distribution function $F$ supported on $\mathbb{R}$. A tolerance interval (TI) is an interval constructed from observed data that is intended to contain a specified proportion of the population with a prescribed level of confidence. Throughout, we focus on univariate tolerance intervals and make no parametric assumptions on $F$.
\subsection{One-Sided Tolerance Intervals}
An upper one-sided tolerance takes the form $(-\infty, U]$, where the bound $U$ is constructed so that 
\begin{align*}
\mathop{\text{Pr}}_{F} ( F(U) \geq P ) \geq 1-\alpha,
\end{align*}
for specified content level $P\in(0,1)$ and confidence level $1-\alpha\in(0,1)$. Equivalently, the bound $U$ is required to exceed the $P$-th population quantile with probability at least $1-\alpha$. Let $Q_P=\inf\{y:F(y)\geq P\}$ denote the $P$-th population quantile. The defining condition for an upper tolerance bound may be written as 
\begin{align*}
\mathop{\text{Pr}}_{F} ( U \geq Q_P ) \geq 1-\alpha.
\end{align*}
thus, the construction of a one-sided tolerance interval reduces to the problem of performing inference on a single population quantile. In particular, any procedure that yields an upper confidence bound for $Q_P$ with coverage $1-\alpha$ immediately induces a valid upper tolerance interval with content $P$. This equivalence highlights the central role of quantile inference in the construction of one-sided tolerance intervals and motivates approaches that target quantile functionals directly.

Analogously, a lower one-sided tolerance interval takes the form $[L,\infty)$, where the bound $L$ is constructed so that
\begin{align*}
\mathop{\text{Pr}}_{F} ( L \leq Q_{1-P} ) \geq 1-\alpha,
\end{align*}
with $Q_{1-P}=\inf\{y:F(y)\geq1-P\}$ the $(1-P)$-th population quantile. The construction of a lower tolerance bound thus reduces to inference on a single lower quantile, mirroring the upper-bound case.
\subsection{Two-Sided Tolerance Intervals}\label{subsec:twosidedTI}
A two-sided tolerance interval takes the form $[L,U]$ and is designed to satisfy
\begin{align*}
    \mathop{\text{Pr}}_{F} ( F(U) - F(L)\geq P ) \geq 1-\alpha.
\end{align*}
Unlike the one-sided case, this condition does not uniquely determine the endpoints $L$ and $U$, and multiple interpretations of two-sided tolerance intervals are possible. A common construction is based on quantile-defined tolerance intervals, in which the endpoints correspond to fixed population quantiles. Specifically, let $0<\tau_L<\tau_U<1$ with $\tau_U-\tau_L = P$, and define 
$L=Q_{\tau_L},\,\,\,\,\,U =Q_{\tau_U}$.
A tolerance interval of this form satisfies the desired content requirement by construction, and frequentist validity is achieved by ensuring joint coverage of the quantile pair $(Q_{\tau_L},Q_{\tau_U})$. Many classical parametric tolerance intervals may be interpreted in this framework. 

An alternative construction is based on content-defined tolerance intervals, in which the interval $[L,U]$ is required only to cover at least a proportion $P$ of the population, without imposing that the endpoints correspond to fixed quantiles. Under this condition, the inferential target is the population content $F(U)-F(L)$, and frequentist validity is assessed directly in terms of this quantity. Classical nonparametric tolerance intervals typically adopt one of these interpretations implicitly, with limited ability to accommodate alternatives. In contrast, the framework developed in this paper allows both quantile-defined and content-defined two-sided tolerance intervals to be constructed within a unified inferential approach.

\subsection{Tolerance Intervals as Quantile Inference Problems}
The preceding discussion shows that tolerance interval construction is fundamentally tied to inference on population quantiles. In the one-sided case, the problem reduces to inference on a single quantile, while in the two-sided case it involves either joint inference on a pair of quantiles or inference on population content expressed in terms of quantiles. This observation motivates a strategy in which tolerance intervals are constructed by first performing inference on the relevant quantile functionals of $F$, and then mapping this uncertainty to tolerance bounds. In the absence of a parametric model for $F$, this requires a nonparametric inferential framework capable of providing uncertainty quantification for quantiles.

In the following sections, we develop such a framework using generalized Bayesian inference via Gibbs posteriors based on the check loss function. This approach yields posterior distributions for population quantiles without requiring a likelihood. By treating the credible limits of these Gibbs posteriors as candidate tolerance bounds, we create a direct map between Bayesian uncertainty quantification and frequentist requirements. Specifically, we show that when the learning rate $\eta$ is properly calibrated, these posterior-based bounds satisfy the formal frequentist coverage conditions for both one-sided and two-sided tolerance intervals.

\section{Gibbs Posteriors for Quantile Inference}
Let $Y_1,...,Y_n$ denote independent and identically distributed (i.i.d.) observations from an unknown distribution $F$ supported on $\mathbb{R}$. We are interested in performing inference on the $\tau$-th population quantile
$Q_\tau=\inf\{y:F(y)\geq\tau\},\,\,\,\,\,0<\tau<1$,
which, as discussed in Section \ref{TI}, is the central object for constructing one- and two-sided tolerance intervals. In the absence of a likelihood for $F$, generalized Bayesian inference \citep{bissiri:2016} provides a principled framework for defining a posterior distribution directly on $Q_\tau$ via a loss function that targets the quantile property.
\subsection{Generalized Bayesian Posterior}
Let $\ell(Q_\tau;y)$ be a loss function measuring the discrepancy between a candidate quantile $Q_\tau$ and an observation $y$. The Gibbs posterior for $Q_\tau$ is defined as being proportional to the exponential of the negative cumulative loss, i.e., 
\begin{align}\label{gibbsposterior}
    \pi(Q_\tau|Y_{1:n})\propto\exp\left\{-\eta\sum_{i=1}^n\ell(Q_\tau;Y_i)\right\}\pi_0(Q_\tau),
\end{align}
where $\pi_0$ is a prior distribution on $Q_\tau$ and $\eta>0$ is a learning rate that determines the concentration of the posterior distribution. Unlike a traditional Bayesian posterior, the Gibbs posterior does not require a likelihood model for the data. By updating the prior through a loss function rather than a density, this framework allows for principled inference directly on specific functionals, such as quantiles, without requiring a complete specification of the data-generating process.

\subsection{Check Loss and Quantile Functionals}
A natural choice for $\ell$ is the check loss (also called the pinball loss):
\begin{align*}
    \rho_\tau(r) =r(\tau-\mathbb{I}\{r<0\}),
\end{align*}
where $r=y-Q_\tau$, and $\mathbb{I}\{\cdot\}$ is the usual indicator function. The check loss is uniquely suited to targeting quantiles since the population minimizer 
$Q^*_\tau = \argmin_{Q} \mathbb{E_F}[\rho_\tau(Y-Q)]$
satisfies $Q^*_\tau=Q_\tau$. Consequently, the Gibbs posterior defined with the check loss concentrates around the $\tau$-th population quantile, providing a nonparametric posterior distribution for this functional.

In the two-sided case, inference on a pair of quantiles $(Q_{\tau_L}, Q_{\tau_U})$ may be obtained by considering the joint Gibbs posterior
\begin{align}\label{twosidedposterior}
    \pi(Q_{\tau_L},Q_{\tau_u}|Y_{1:n})\propto \exp\left\{-\eta\sum_{i=1}^n[\rho_{\tau_l}(Y_i-Q_{\tau_l})+\rho_{\tau_u}(Y_i-Q_{\tau_u})]\right\}\pi_0(Q_{\tau_l},Q_{\tau_u}). 
\end{align}
This joint formulation allows us to simultaneously characterize the uncertainty of both interval endpoints, accommodating both the quantile-defined and content-defined interpretations introduced in Section \ref{subsec:twosidedTI}.

\subsection{Construction of Tolerance Intervals}\label{subsec:BayesTI}
The Gibbs posterior provides a representation of uncertainty for the quantile functionals, which is used to map to the tolerance bounds $L$ and $U$. Since one-sided tolerance intervals reduce to inference on a single quantile, we define the upper bound $U$ as the $(1-\alpha)$ posterior quantile of $\pi(Q_P \mid Y_{1:n})$. Similarly, the lower bound $L$ is defined as the $\alpha$ posterior quantile of $\pi(Q_{1-P} \mid Y_{1:n})$. This construction ensures that, under proper calibration of the learning rate $\eta$, these bounds satisfy the required frequentist confidence level $1-\alpha$ relative to the target population quantiles.

For two-sided tolerance intervals, the construction is more involved as it requires summarizing the joint posterior $\pi(Q_{\tau_L}, Q_{\tau_U} \mid Y_{1:n})$. It is important to note that for two-sided intervals, one cannot simply use the marginal posterior quantiles of $Q_{\tau_L}$ and $Q_{\tau_U}$ as the tolerance bounds. While such an approach might provide a credible interval for the quantiles themselves, it does not guarantee that the resulting interval $[L, U]$ will satisfy the frequentist coverage or content requirements defined in Section \ref{TI}. Because the two-sided tolerance condition depends on the joint behavior of the interval endpoints, we adopt a decision rule based on symmetry, first described by \cite{wolfinger:1998} and later refined by \cite[Chapter 11]{krishnamoorthy2009statistical}, to summarize the joint posterior distribution $\pi(Q_{\tau_L}, Q_{\tau_U} \mid Y_{1:n})$ and ensure that the resulting interval is sufficiently wide to achieve the prescribed confidence level. To accomplish this, we first compute the posterior mean of the midpoints, $\bar{\mu} = \mathbb{E}[(Q_{\tau_L} + Q_{\tau_U})/2 \mid Y_{1:n}]$, and then define the upper bound $U$ as the $(1-\alpha)$ posterior quantile of the random variable $U^* = \max(Q_{\tau_U}, 2\bar{\mu} - Q_{\tau_L})$. The lower bound is then set as $L = 2\bar{\mu} - U$. This ensures the resulting tolerance interval is centered relative to the posterior midpoints of the targeted quantiles.

Figure~\ref{fig:twosidedconstruction} illustrates the construction of a two-sided tolerance interval using the joint Gibbs posterior distribution of the lower and upper quantiles $(Q_{\tau_L},Q_{\tau_U})$ for a nominal confidence level of $1-\alpha=0.90$. Each point represents a posterior draw of the quantile pair. The solid boundary corresponds to the symmetry-based decision rule used to summarize the joint posterior, while the dashed lines indicate the bounds obtained by applying marginal posterior quantiles to $Q_{\tau_L}$ and $Q_{\tau_U}$ separately.

The figure makes clear that marginal posterior intervals fail to account for the dependence between the lower and upper quantiles and therefore exclude a substantial fraction of posterior draws that are jointly plausible. In contrast, the joint construction explicitly incorporates this dependence, producing an interval that is sufficiently wide to meet the prescribed confidence requirement. In this example, approximately 90\% of joint posterior draws fall within the symmetry-based region, matching the nominal confidence level, whereas only about 81\% are retained under the marginal construction. The symmetry line highlights how the final bounds are centered relative to the posterior midpoints of the targeted quantiles, ensuring balanced expansion of the interval endpoints.

\begin{figure}[htbp]
\centering
\includegraphics[width=0.7\textwidth]{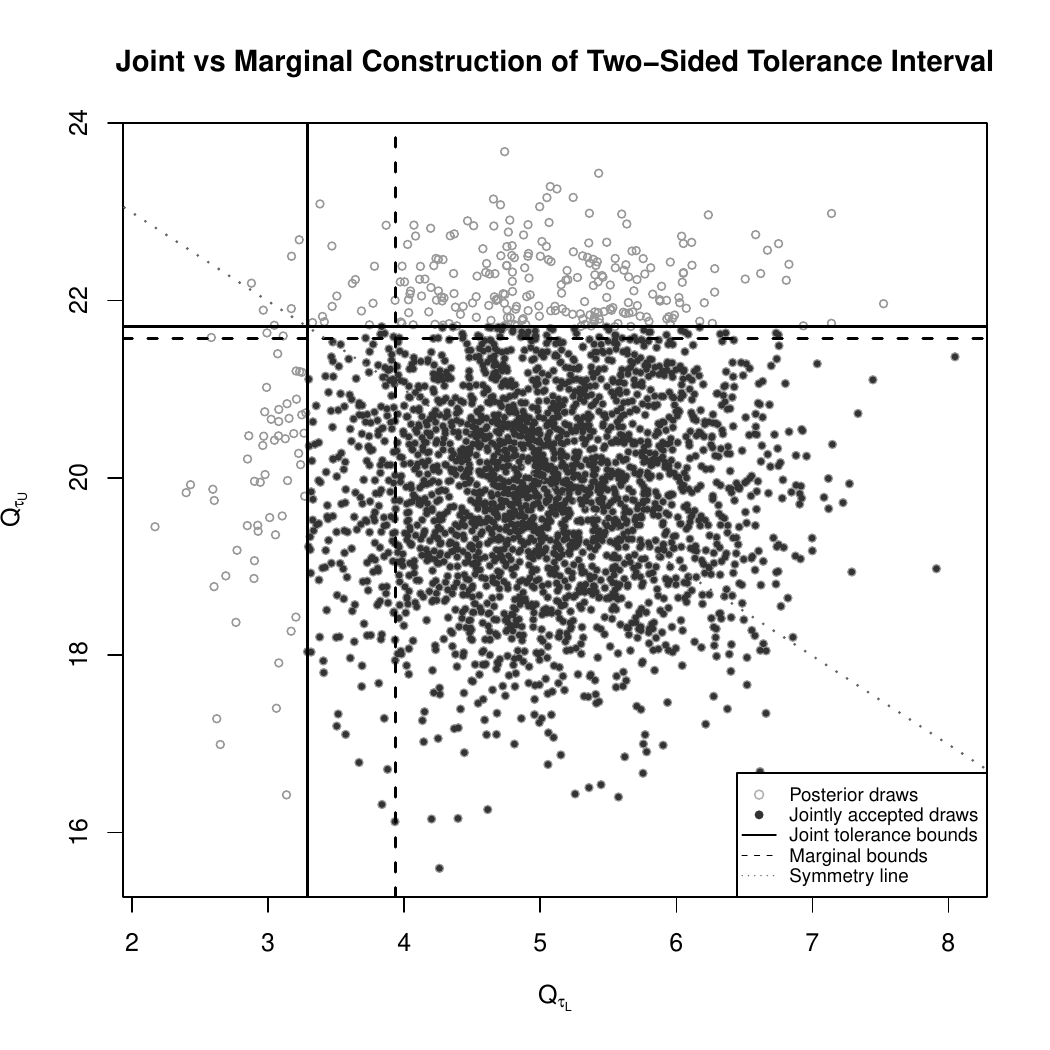}
\caption{Joint Gibbs posterior draws of $(Q_{\tau_L},Q_{\tau_U})$ and the symmetry-based two-sided tolerance interval construction.}
\label{fig:twosidedconstruction}
\end{figure}

\subsection{Prior Specification and Posterior Computation}
Our approach permits a fully nonparametric analysis in which inference is driven entirely by the loss function, without the need to posit a parametric sampling model or to inject prior information on the target functional. Accordingly, unless genuine prior information about the quantiles is available, we adopt a flat (improper) prior on the quantile functional(s), i.e., $\pi_0(Q_\tau)\propto 1$. This choice reflects an intentional lack of prior structure rather than a simplifying assumption, and ensures that posterior inference is determined by the data through the check loss and the calibrated learning rate $\eta$. While we focus on flat priors to ensure an objective, data-driven analysis, the Gibbs framework readily accommodates informative priors if such information is available. For instance, a proper normal prior could be placed on $Q_\tau$ to reflect expert knowledge, or a Uniform$(a, b)$ prior could be used to restrict the quantiles to a physically plausible range. In such cases, the prior $\pi_0(Q_\tau)$ would simply be multiplied by the exponential loss term in (\ref{gibbsposterior}) or (\ref{twosidedposterior}) without altering the underlying calibration logic.

In the two-sided setting, inference targets an ordered pair of quantiles. To enforce the constraint that the upper quantile, $Q_{\tau_U}$, exceeds the lower quantile, $Q_{\tau_L}$, while maintaining an unconstrained parameter space for computation, we reparameterize the posterior distribution in (\ref{twosidedposterior}) as follows:
\begin{align}
    \theta_1 &= Q_{\tau_L}\\
    \theta_2 &= \log(Q_{\tau_U}-Q_{\tau_L}),
\end{align}
so that $Q_{\tau_L}= \theta_1$ and $Q_{\tau_U}=\theta_1 + e^{\theta_2}$. This transformation guarantees that $Q_{\tau_U}$ will always be greater than $Q_{\tau_L}$. The resulting transformed posterior is then
\begin{align}
p(\theta_1,\theta_2 \mid Y_{1:n})
\;\propto\;
\exp\Bigg\{
-\eta \Bigg[
\sum_{i=1}^n \rho_{\tau_L}(Y_i - \theta_1)
+
\sum_{i=1}^n \rho_{\tau_U}\!\left(Y_i - \big(\theta_1 + e^{\theta_2}\big)\right)
\Bigg]
\Bigg\}
\times e^{\theta_2}.
\end{align}
The term $e^{\theta_2}$ in the resulting density accounts for the Jacobian of this transformation, ensuring the posterior remains valid in the unconstrained space. Posterior sampling is performed in the unconstrained parameterization using standard Markov chain Monte Carlo methods, such as random-walk Metropolis–Hastings or slice sampling. Posterior draws are transformed back to the original quantile scale, where the joint samples $(\theta_1^{(k)}, \theta_1^{(k)} + e^{\theta_2^{(k)}})$ are used to calculate the midpoints $\bar{\mu}$ and the auxiliary variable $U^*$ required to form the tolerance intervals as described in Section \ref{subsec:BayesTI}. 

\section{Calibration of the Learning Rate $\eta$}\label{sec:calibration}
The Gibbs posterior is defined through a loss function rather than a likelihood; consequently, its posterior dispersion depends critically on the learning rate $\eta$. While the posterior mode targets the desired quantile functional regardless of the value of $\eta$, the resulting credible intervals may be either too narrow or too wide when interpreted as tolerance intervals. Calibration of $\eta$ is therefore required to ensure that posterior-based tolerance intervals achieve nominal frequentist coverage.

We adopt a generalized posterior calibration (GPC) strategy \citep{martin:2022} that selects $\eta$ to directly enforce the defining coverage property of the desired interval. In this framework, $\eta$ is explicitly tuned to align the posterior concentration with the sampling variability of the target functional, thereby ensuring that the Bayesian credible intervals function as valid frequentist tolerance intervals. We distinguish between two primary calibration objectives: (i) Quantile Calibration, where the goal is valid frequentist coverage of the population functionals (quantiles), and (ii) Content Calibration, where the goal is to satisfy the $P$-content requirement of a tolerance interval with probability $1-\alpha$.
\subsection{The Calibration Procedure}
Since the true data-generating distribution is unknown, we calibrate the learning rate $\eta$ to ensure the Gibbs posterior delivers valid frequentist performance. We seek the optimal learning rate $\eta^*$ that satisfies the root of the coverage error function:$$h(\eta) = \mathbb{E}_{F}[\1(\text{Success})] - (1-\alpha) = 0$$where ``Success" is defined by the specific calibration objective (Quantile or Content). In practice, because $F$ is unknown and the expectation is evaluated via bootstrap sampling, $h(\eta)$ is observed with noise. To solve this, we employ a stochastic approximation approach via the Robbins-Monro algorithm \citep{robbins:1951}. Given a target confidence level $1-\alpha$, the learning rate is updated iteratively:\begin{align}\eta_{s+1} = \eta_s + \kappa_s \left( \hat{C}(\eta_s) - (1-\alpha) \right),\end{align}where $\hat{C}(\eta_s)$ is the estimated coverage probability obtained via $B$ bootstrap replicates at the current learning rate, and $\kappa_s = c / (1+s)^\gamma$ is a decaying step-size sequence. The hyperparameters $c$ and $\gamma$ are chosen to satisfy the standard conditions for stochastic approximation, ensuring that $\eta_s$ converges to the optimal learning rate $\eta^*$ as $s \to \infty$. This approach allows for robust calibration even when the coverage function is evaluated with Monte Carlo error. In our implementation, we use the sample quantiles of the full dataset as the ``true" parameters during the bootstrap process to evaluate the success of each iteration.

Specifically, at each iteration $s$ of the Robbins-Monro algorithm, the estimated coverage $\hat{C}(\eta_s)$ is computed through the following bootstrap procedure:
\begin{enumerate}
    \item Bootstrap Resampling: Generate $B$ bootstrap datasets $Y^{(1)}_{1:n}, \dots, Y^{(B)}_{1:n}$ by sampling with replacement from the original data $Y_{1:n}$.
    \item Posterior Sampling: For each bootstrap dataset $b$, construct the Gibbs posterior $\pi(Q_\tau \mid Y^{(b)}_{1:n}, \eta_s)$ and obtain posterior draws (via slice sampling or Metropolis-Hastings).
    \item Interval Construction: Using these draws, construct the $(1-\alpha)$ tolerance interval (or bound) $[L^{(b)}, U^{(b)}]$ according to the methods described in Section \ref{subsec:BayesTI}.
    \item Success Evaluation: Determine if the interval satisfies the specific success criterion associated with the inferential objective (see Sections \ref{quantilecalibration} and \ref{contentcalibration}).
    \item Aggregation: The estimate $\hat{C}(\eta_s)$ is the proportion of the $B$ replicates that resulted in a success.
\end{enumerate}
By using the original empirical quantiles (or the original data set) as the ``truth" in Step 4, we ensure that the learning rate $\eta$ is calibrated to the specific features of the observed distribution without assuming a parametric form.

\subsection{Practical Considerations}\label{subsec:considerations}
While the GPC procedure is theoretically justified under standard stochastic approximation conditions, successful implementation requires careful attention to tuning parameters and convergence diagnostics. A detailed discussion of these practical aspects, including recommendations for initial values, step-size selection, and monitoring convergence, is provided in Section 1 of the Supplementary Material.

\subsection{Quantile-Defined Calibration}\label{quantilecalibration}
In the quantile-defined setting, the objective is to calibrate $\eta$ such that the Gibbs posterior credible intervals provide valid frequentist coverage for the target population functionals. In our implementation, ``Success" at each bootstrap replicate $b$ is defined by whether the interval captures the ``true" quantile (estimated by the empirical quantile of the original data, $\hat{Q}_\tau$).
\begin{itemize}
    \item One-Sided Upper Bound: Success is defined as $\hat{Q}_P \leq U^{(b)}$, where $U^{(b)}$ is the $(1-\alpha)$ posterior quantile of $\pi(Q_P \mid Y^{(b)}_{1:n}, \eta_s)$.
    \item One-Sided Lower Bound: Success is defined as $\hat{Q}_{1-P} \geq L^{(b)}$, where $L^{(b)}$ is the $\alpha$ posterior quantile of $\pi(Q_{1-P} \mid Y^{(b)}_{1:n}, \eta_s)$.
    \item Two-Sided Joint Intervals: Success is defined as the event where the pair of population quantiles $(\hat{Q}_{\tau_L}, \hat{Q}_{\tau_U})$ falls within the joint credible region constructed via the symmetry rule described in Section \ref{subsec:BayesTI}.
\end{itemize}
Note, for one-sided bounds, quantile calibration and content calibration are mathematically equivalent. For example, requiring that the upper bound $U$ covers the $P$-th quantile ($\hat{Q}_P \leq U$) is identical to requiring that the proportion of the population below $U$ is at least $P$ (i.e., $F(U) \geq P$). This equivalence does not hold for two-sided intervals, necessitating the distinct calibration approach described in Section \ref{subsec:BayesTI}.

To formalize the implications of this calibration target, we next state the asymptotic validity result for the one-sided procedure. The theorem shows that once $\eta$ is calibrated to the one-sided success event, the resulting posterior tolerance bound is asymptotically correct from a frequentist coverage perspective.

\begin{theorem}[One-sided asymptotic validity of calibrated Gibbs bounds]\label{thm:onesided-main}
Under standard regularity conditions for quantile inference (continuity and positivity of $f$ at $Q_\tau$, locally positive prior, compact calibration bracket for $\eta$, and uniform consistency of the estimated coverage curve), the calibrated learning rate $\hat\eta_n$ converges in probability to a unique target value $\eta^\star$, and the resulting one-sided calibrated Gibbs tolerance bound attains asymptotic nominal coverage:
\[
P_F\{Q_\tau\le U_n(\hat\eta_n)\}=1-\alpha+o_p(1).
\]
In particular, for the quantile-calibration target used here,
\[
\eta^\star=\frac{f(Q_\tau)}{\tau(1-\tau)}.
\]
\end{theorem}
\noindent\emph{Proof.} See the supplementary material.

\subsection{Content-Defined Calibration}\label{contentcalibration}
For two-sided tolerance intervals, the primary interest is often the content of the resulting interval rather than the precision of the specific quantile estimates. Under this objective, we seek a learning rate $\eta$ that ensures the interval $[L, U]$ contains at least a proportion $P$ of the population with probability $1-\alpha$. Because the true distribution function $F$ is unknown, we estimate success by checking if the interval derived from the bootstrap sample captures the required proportion of the original data $Y_{1:n}$. This is equivalent to evaluating the success criterion against the empirical cumulative distribution function (ECDF) of the observed data, denoted $F_n$. Specifically, we define:
\begin{equation*}
\text{Success} \iff F_n(U^{(b)}) - F_n(L^{(b)}) \geq P,
\end{equation*}
which can be expressed computationally as:
\begin{equation*}
\text{Success} \iff \frac{1}{n} \sum_{i=1}^n \1(L^{(b)} \leq Y_i \leq U^{(b)}) \geq P.
\end{equation*}
By using the original data as a proxy for the population ``truth", the Robbins-Monro algorithm tunes $\eta$ so that the posterior width is sufficient to satisfy the content requirement $P$ across the bootstrap replicates. It is important to emphasize that this calibration often yields a different optimal $\eta^*$ than the quantile-defined approach. While quantile calibration focuses on the location and coverage of specific functionals (the points $Q_{\tau_L}$ and $Q_{\tau_U}$), content calibration is a ``global" requirement on the total mass captured by the interval width. 

Building directly on this distinction, the following theorem summarizes the two-sided large-sample guarantees under both calibration objectives. 

\begin{theorem}[Two-sided asymptotic validity of calibrated Gibbs intervals]\label{thm:twosided-main}
Under the corresponding two-quantile regularity conditions (continuity and positivity of $f$ at $Q_{\tau_L}$ and $Q_{\tau_U}$, locally positive prior for $(Q_{\tau_L},Q_{\tau_U})$, compact calibration bracket, and uniform consistency of estimated coverage), the calibrated learning rate converges to a unique target under either calibration criterion:
\begin{itemize}
    \item quantile-defined calibration, where success is $\{L_n\le Q_{\tau_L},\ U_n\ge Q_{\tau_U}\}$;
    \item content-defined calibration, where success is $\{F(U_n)-F(L_n)\ge P\}$.
\end{itemize}
Consequently, the resulting two-sided calibrated Gibbs interval achieves asymptotic nominal coverage for the selected criterion:
\[
P_F\big(\text{Success at }\hat\eta_n\big)=1-\alpha+o_p(1).
\]
\end{theorem}
\noindent\emph{Proof.} See the supplementary material.

\section{Simulations}\label{sec:simulations}
We compare the proposed calibrated Gibbs (Cal-Gibbs) tolerance intervals with established nonparametric and Bayesian alternatives commonly used for quantile-based inference and tolerance limit construction. Full methodological details, derivations, and prior specifications are provided in the supplementary metrials.

\subsection{Comparative Methods}\label{subsec:comparative}
We evaluate the Calibrated Gibbs (Cal-Gibbs) posterior against several benchmarks. The classical Wilks method \citep{wilks:1941} defines limits using order statistics. However, its discrete nature often results in conservative coverage and unnecessarily wide intervals. To address this, \citep{young:2014} propose interpolated and extrapolated order statistics (YM) to achieve a coverage closer to the nominal levels. We also compare against two Bayesian approaches: Bayesian Quantile Regression (BQR-AL) \citep{yu:2001} and the Extended Asymmetric Laplace (Ext-AL) model \citep{yan:2025}. These methods utilize the Asymmetric Laplace distribution as a working likelihood. While flexible, these models are inherently misspecified if the data-generating process does not follow an AL distribution. Without an explicit calibration step like the one proposed in our Cal-Gibbs framework, these models may fail to maintain reliable frequentist coverage in heavy-tailed or highly asymmetric settings. 

\subsection{One-Sided Tolerance Interval  Results}\label{sec:onesidedsimulations}
To evaluate the empirical performance of the proposed calibrated Gibbs posterior, we conduct a comprehensive simulation study across four distributional settings with varying characteristics. We focus specifically on upper one-sided $(P, 1-\alpha)$ tolerance limits, as the results for the lower one-sided case were found to be symmetric in performance. To assess the ability of the method to handle tail estimation (which is inherently more challenging than center-of-mass estimation),we vary the content level over $P \in \{0.90, 0.95, 0.99\}$ while fixing the confidence level at $1-\alpha = 0.90$ across all trials.

A critical consideration in this simulation design is the sample size $n$. For the classical Wilks method, even the most optimistic choice of the upper tolerance limit, $U = Y_{(n)}$, requires the sample size to satisfy $\text{Pr}\{F(Y_{(n)}) \ge P\} = 1 - P^n \ge 1-\alpha.$ Solving for $n$ yields the minimum required sample size $n \ge \ln(\alpha)/\ln(P)$. For example, achieving a $99\%$ content level at $90\%$ confidence requires at least $n = 230$ observations. When $n$ falls below this threshold, the Wilks method cannot mathematically guarantee the required coverage, even when using the most extreme observation. Consequently, we select sample sizes $n \in \{22, 45, 230\}$ corresponding to $P \in \{0.90, 0.95, 0.99\}$.

For each experimental setting, we perform 1{,}000 independent Monte Carlo repetitions. Performance is evaluated using two primary metrics: empirical coverage, defined as the proportion of trials in which the estimated upper limit $U$ satisfies $F(U) \ge P$, and the average length of the resulting tolerance bound.

We consider four distributions chosen to represent a broad spectrum of challenges for quantile-based uncertainty quantification: (i) a standard normal, $N(0,1)$, as a symmetric, light-tailed baseline; (ii) a Gamma$(2,1)$ to introduce moderate skewness; (iii) a heavy-tailed Pareto$(1,2)$ to highlight the risks of misspecified working likelihoods in BQR-AL; and (iv) a highly leptokurtic Normal Mixture, $0.9\mathcal{N}(0,1) + 0.1\mathcal{N}(0,10^2)$, which tests the stability of the calibrated learning rate $\eta$ against extreme outliers. These settings allow us to identify where traditional Bayesian credible intervals fail to provide frequentist tolerance guarantees and demonstrate how the calibrated Gibbs posterior adapts to maintain nominal $1-\alpha$ coverage across diverse distributional shapes and tail behaviors.

The results of the one-sided simulation study are summarized in Table~\ref{tab:ti_sim_nox}. Several clear patterns emerge regarding both the reliability of frequentist coverage and the efficiency of the resulting tolerance limits across the different distributional settings. The proposed Cal-Gibbs method exhibits remarkably stable performance, maintaining empirical coverage close to the nominal $0.90$ level across all four distributions and all three content levels. In contrast, the Bayesian benchmark methods display substantial instability. The BQR-AL model with normal approximation is notably conservative in the light-tailed Normal and Gamma settings, but suffers from severe under-coverage in the heavy-tailed Pareto ($0.694$) and Mixture ($0.643$) cases at $P=0.99$. This behavior confirms that a fixed working likelihood (even when paired with an asymptotic sandwich correction) fails to adequately capture uncertainty when the tail behavior of the data-generating process exceeds that of the Asymmetric Laplace distribution. The Ext-AL model exhibits even greater volatility. Although it occasionally attains high coverage, its performance is highly inconsistent across settings, particularly at $P=0.95$, where empirical coverage falls as low as $0.367$. This indicates that while the Extended AL distribution provides additional flexibility, the absence of an explicit frequentist calibration mechanism renders it unreliable for tolerance interval construction.

As expected, the nonparametric Wilks and YM methods generally achieve correct coverage whenever the sample size satisfies the feasibility condition $n \ge \ln(\alpha)/\ln(P)$. However, this robustness comes at the cost of substantially wider tolerance limits. In the Pareto and Mixture settings, the Cal-Gibbs method produces intervals that are markedly shorter than those of the nonparametric benchmarks. For instance, in the Pareto case at $P=0.90$, Cal-Gibbs yields an average interval length of $4.877$, compared to $9.418$ for Wilks and $9.261$ for the YM method. This gain in efficiency arises from the Gibbs posterior’s ability to incorporate information from the full sample through the calibrated learning rate $\eta$. In contrast, the Wilks and YM procedures rely solely on extreme order statistics, which are inherently unstable in heavy-tailed distributions. The Cal-Gibbs approach therefore achieves a more concentrated and stable estimate of the target quantile while still maintaining the required frequentist coverage.

\begin{table}[htbp]
\centering
{\tiny
\begin{tabular}{l l | ccc | ccc}
\hline
& & \multicolumn{6}{c}{Content level $P$} \\
\cline{3-8}
& 
& \multicolumn{3}{c|}{Coverage} 
& \multicolumn{3}{c}{Length} \\
\cline{3-8}
Distribution & Method 
& $P=0.90$ & $P=0.95$ & $P=0.99$
& $P=0.90$ & $P=0.95$ & $P=0.99$ \\
\hline

\multirow{5}{*}{$\mathcal{N}(0,1)$}
 & Cal-Gibbs & 0.896 & 0.905 & 0.905 & 1.733 & 2.030 & 2.656 \\
 & Wilks     & 0.886 & 0.907 & 0.899 & 1.921 & 2.213 & 2.784 \\
 & YM        & 0.902 & 0.899 & 0.906 & 1.898 & 2.205 & 2.772 \\
 & BQR-AL    & 0.997 & 1.000 & 1.000 & 2.208 & 2.545 & 3.196 \\
 & Ext-AL    & 0.873 & 0.716 & 0.998 & 1.654 & 1.790 & 12.237 \\
\hline

\multirow{5}{*}{Gamma$(2,1)$}
 & Cal-Gibbs & 0.895 & 0.893 & 0.901 & 4.916 & 5.776 & 7.661 \\
 & Wilks     & 0.900 & 0.890 & 0.894 & 5.544 & 6.354 & 8.241 \\
 & YM        & 0.901 & 0.898 & 0.902 & 5.531 & 6.345 & 8.239 \\
 & BQR-AL    & 0.943 & 0.933 & 0.936 & 5.082 & 5.885 & 7.745 \\
 & Ext-AL    & 0.679 & 0.367 & 1.000 & 4.256 & 4.592 & 35.507 \\
\hline

\multirow{5}{*}{Pareto$(1,2)$}
 & Cal-Gibbs & 0.899 & 0.898 & 0.893 & 4.877 & 7.211 & 17.082 \\
 & Wilks     & 0.893 & 0.903 & 0.902 & 9.418 & 12.892 & 24.822 \\
 & YM        & 0.902 & 0.903 & 0.894 & 9.261 & 12.875 & 24.793 \\
 & BQR-AL    & 0.942 & 0.875 & 0.694 & 4.807 & 6.473 & 12.360 \\
 & Ext-AL    & 0.638 & 0.399 & 0.968 & 3.826 & 4.538 & 17.714 \\
\hline

\multirow{5}{*}{$0.9\mathcal{N}(0,1)+0.1\mathcal{N}(0,10^2)$}
 & Cal-Gibbs & 0.892 & 0.908 & 0.899 & 3.976 & 6.748 & 18.802 \\
 & Wilks     & 0.903 & 0.899 & 0.898 & 8.323 & 11.832 & 20.262 \\
 & YM        & 0.899 & 0.899 & 0.906 & 8.295 & 11.820 & 20.244 \\
 & BQR-AL    & 0.989 & 0.959 & 0.643 & 4.446 & 6.407 & 15.066 \\
 & Ext-AL    & 0.894 & 0.736 & 0.886 & 3.678 & 4.321 & 18.371 \\
\hline
\end{tabular}
}
\caption{Empirical coverage and average interval length for one-sided upper tolerance intervals at confidence level $1-\alpha=0.90$, across content levels $P$.}
\label{tab:ti_sim_nox}
\end{table}

To further illustrate the robustness of the calibrated Gibbs posterior, a sensitivity analysis was conducted by varying the sample size incrementally between $n=15$ and $n=30$ for an upper one-sided tolerance limit with $P=0.9$ and $1-\alpha=0.9$. This specific range was chosen to bracket the theoretical minimum sample size of $n=22$ required by the Wilks method for this configuration. The empirical coverage probabilities across these sample sizes reveal a stark contrast in how the different methods respond to small-sample constraints (Figure \ref{fig:both}). The Cal-Gibbs method remains consistently calibrated near the nominal $0.90$ level across the entire range of sample sizes, maintaining stability even when $n < 22$. In contrast, both the Wilks and YM methods exhibit significant coverage fluctuations, systematically under-covering for sample sizes below the theoretical threshold. It is only as the sample size approaches and exceeds $n=22$ that the nonparametric methods achieve the nominal coverage, eventually becoming notably conservative as the sample size increases toward $n=30$.

\begin{figure}[htbp]
\centering
\includegraphics[width=0.47\textwidth]{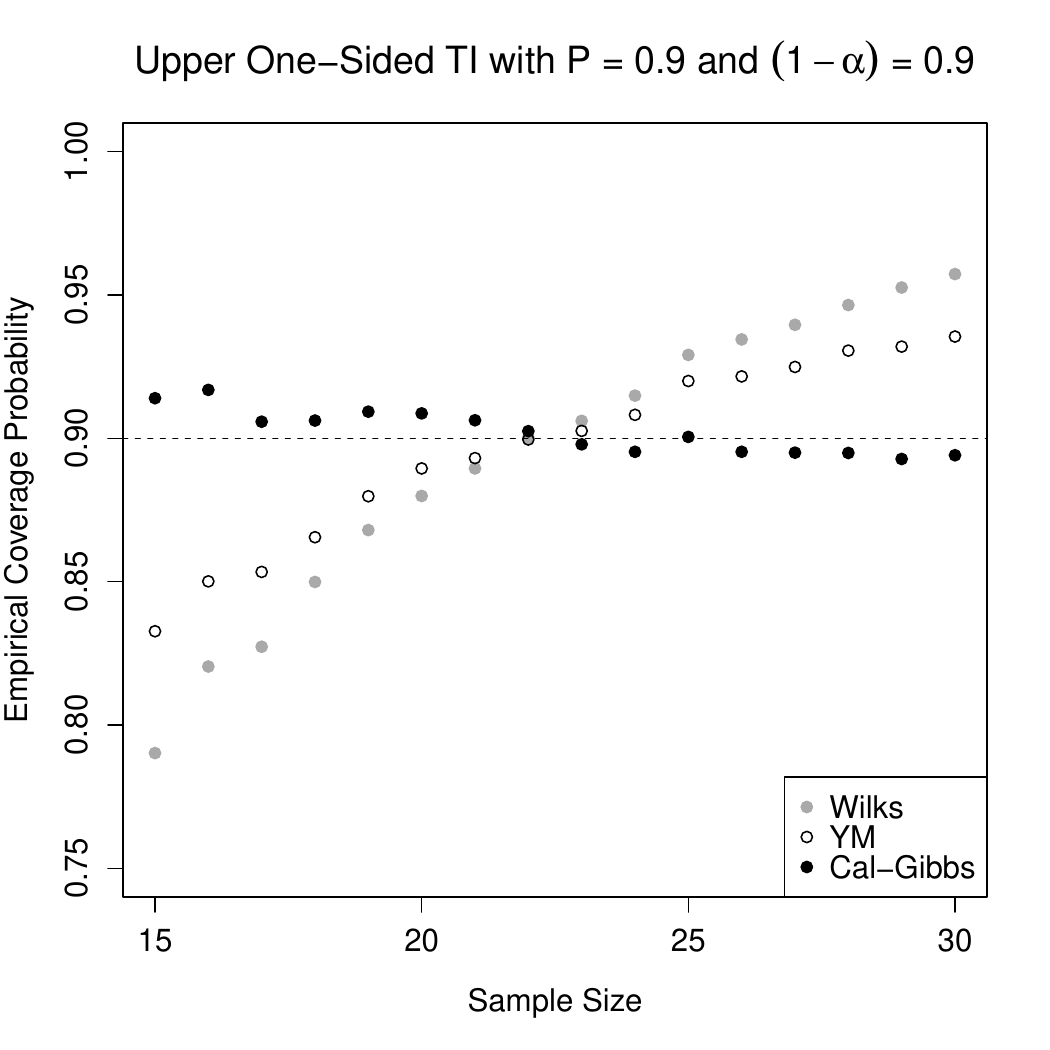}
\hfill
\includegraphics[width=0.47\textwidth]{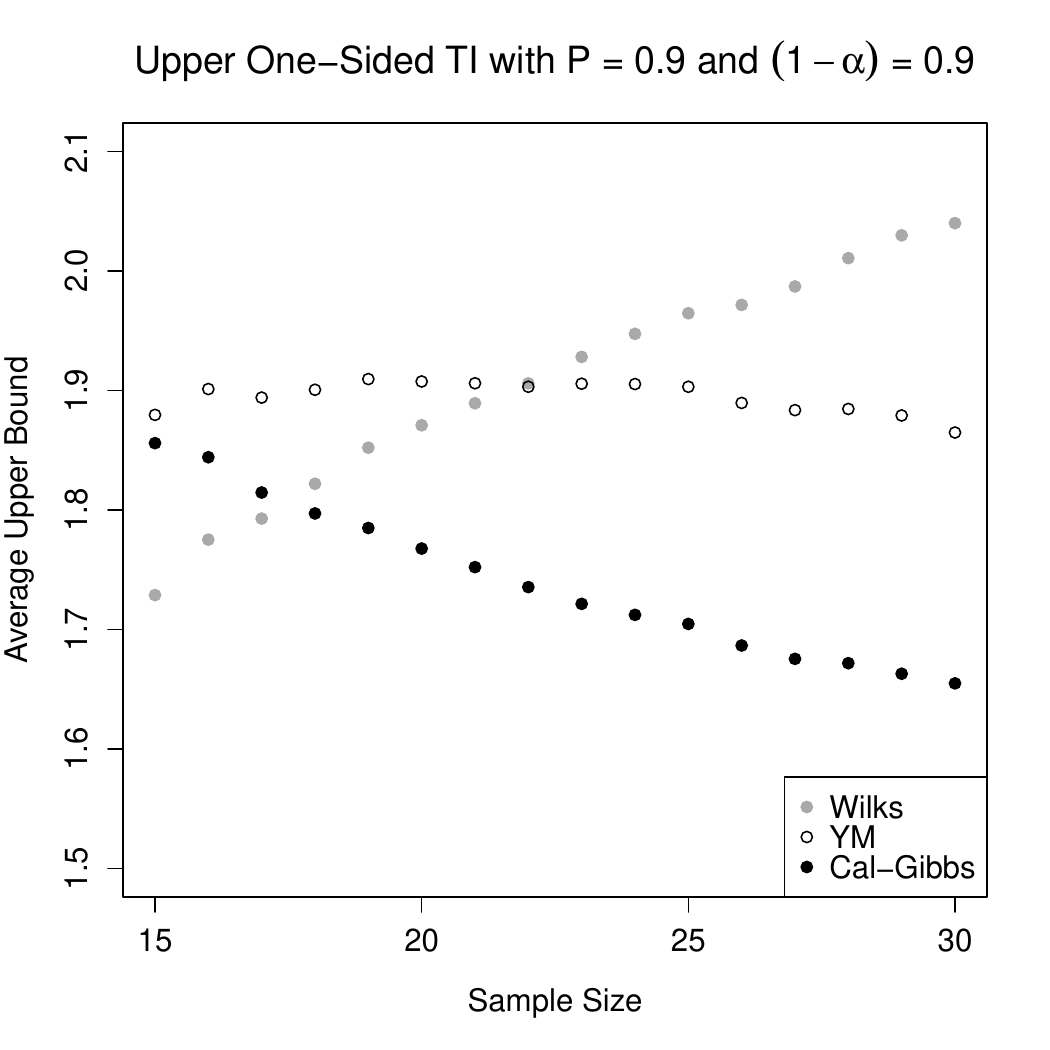}
\caption{Empirical performance of upper one-sided tolerance intervals with content level $P=0.9$ and confidence level $1-\alpha=0.9$ as a function of sample size. Left panel: empirical coverage probability, with the dashed horizontal line indicating the nominal confidence level. Right panel: average upper tolerance bound. Results are shown for the Wilks nonparametric method, the YM method, and the calibrated Gibbs posterior approach.}
\label{fig:both}
\end{figure}

The efficiency of these limits is further clarified by examining the average upper bounds over the same range of sample sizes. Despite maintaining higher and more consistent coverage than the nonparametric benchmarks at the smallest sample sizes, the Cal-Gibbs method simultaneously produces tighter limits. For sample sizes where Wilks and YM are under-covering, the Cal-Gibbs method provides valid coverage while yielding average upper bounds that are lower than the non-conservative YM method. As the sample size increases beyond the critical threshold of $n=22$, the Cal-Gibbs average upper bound continues to decrease and remains significantly lower than the bounds produced by the Wilks method, which grow increasingly wide as they over-cover. By $n=30$, the Cal-Gibbs bound is approximately $1.65$, whereas the Wilks bound exceeds $2.0$. This analysis confirms that the calibrated Gibbs posterior is not only more robust to small sample sizes where traditional order statistics fail, but it is also more efficient, providing narrower and more accurate tolerance limits while strictly maintaining the required frequentist confidence level.

Given that BQR-AL and Ext-AL failed to provide reliable coverage in the preceding one-sided stress tests, we exlcude them from the two-sided simulations (Section \ref{sec:twosidedsimulations}) to focus on the comparison between the proposed calibrated Gibbs method and the established nonparametric standards.

\subsection{Two-Sided Tolerance Interval  Results}\label{sec:twosidedsimulations}
For the two-sided tolerance interval evaluations, we maintain the same experimental rigor, conducting 1,000 simulations while fixing the confidence level at $1-\alpha = 0.90$ and varying the content level $P \in \{0.90, 0.95, 0.99\}$. The analysis focuses on two distinct distributional settings: a symmetric Normal distribution, $\mathcal{N}(10,3)$, and an asymmetric Beta$(5,2)$ distribution. Similar to the one-sided case, the choice of sample size $n$ is critical for the existence of a valid two-sided nonparametric limit. For a two-sided tolerance interval to exist using the sample extremes $Y_{(1)}$ and $Y_{(n)}$, the sample size must satisfy the condition $(n - 1)P^n - nP^{n-1} + 1 \ge 1-\alpha$. This requirement results in minimum sample sizes of $n = 38, 77,$ and $388$ corresponding to $P = 0.90, 0.95,$ and $0.99$, respectively. We utilize these thresholds in our simulation design to ensure the Wilks method is evaluated under conditions where it can mathematically achieve the target coverage.

A significant distinction in the two-sided evaluation is the consideration of two different coverage definitions: content-defined coverage (Section \ref{contentcalibration}) and quantile-defined coverage (Section \ref{quantilecalibration}). The results, summarized in Table \ref{tab:ti_sim_twosided}, reveal how the benchmarks perform under these competing criteria.

\begin{table}[htbp]
\centering
{\tiny
\begin{tabular}{l l l | ccc | ccc}
\hline
& & & \multicolumn{6}{c}{Content level $P$} \\
\cline{4-9}
& & & \multicolumn{3}{c|}{Coverage} & \multicolumn{3}{c}{Length} \\
Distribution & Method & Coverage type & $P=0.90$ & $P=0.95$ & $P=0.99$ & $P=0.90$ & $P=0.95$ & $P=0.99$ \\
\hline
\multirow{6}{*}{$\mathcal{N}(10,3)$}
 & Wilks & Content  & 0.899 & 0.900 & 0.900 & 12.785 & 14.442 & 17.783 \\
 & YM    & Content  & 0.897 & 0.898 & 0.899 & 12.731 & 14.414 & 17.780 \\
 & Cal-Gibbs & Content & 0.900 & 0.897 & 0.896 & 11.950 & 13.626 & 16.995 \\
 \cline{2-9}
 & Wilks & Quantile & 0.726 & 0.731 & 0.744 & 12.785 & 14.442 & 17.783 \\
 & YM    & Quantile & 0.724 & 0.729 & 0.744 & 12.731 & 14.414 & 17.780 \\
 & Cal-Gibbs & Quantile & 0.908 & 0.899 & 0.891 & 13.536 & 14.901 & 17.944 \\
\hline
\multirow{6}{*}{Beta$(5,2)$}
 & Wilks & Content  & 0.900 & 0.909 & 0.905 & 0.625 & 0.685 & 0.786 \\
 & YM    & Content  & 0.894 & 0.907 & 0.904 & 0.623 & 0.684 & 0.786 \\
 & Cal-Gibbs & Content & 0.898 & 0.899 & 0.906 & 0.601 & 0.664 & 0.780 \\
 \cline{2-9}
 & Wilks & Quantile & 0.729 & 0.734 & 0.736 & 0.625 & 0.685 & 0.786 \\
 & YM    & Quantile & 0.729 & 0.731 & 0.735 & 0.623 & 0.684 & 0.786 \\
 & Cal-Gibbs & Quantile & 0.890 & 0.904 & 0.901 & 0.673 & 0.744 & 0.837 \\
\hline
\end{tabular}
}
\caption{Empirical coverage and average interval length for two-sided tolerance intervals at confidence level $1-\alpha=0.90$, across content levels $P$.}
\label{tab:ti_sim_twosided}
\end{table}

The results for content-defined coverage show that the Wilks, YM, and Cal-Gibbs methods all maintain empirical coverage near the nominal $0.90$ level across both symmetric and asymmetric distributions. However, the Cal-Gibbs method consistently achieves this coverage with shorter average interval lengths. For example, in the $\mathcal{N}(10,3)$ setting with $P=0.90$, the Cal-Gibbs method produces an average length of $11.950$, whereas the Wilks and YM methods yield $12.785$ and $12.731$, respectively.

The most striking differences appear when examining quantile-defined coverage. While the Cal-Gibbs method successfully maintains coverage probabilities near $0.90$ for both distributions, the performance of the Wilks and YM methods degrades significantly, with coverage levels dropping to approximately $0.73$ in both the Normal and Beta settings. This sharp decline is not entirely unexpected, as traditional nonparametric methods like Wilks and YM were specifically designed to satisfy content-based requirements rather than to provide coverage guarantees for specific population quantiles. However, this highlights a major practical limitation in that these methods lack the structural flexibility to adapt to different inferential goals. In contrast, the calibrated Gibbs posterior provides a unified framework that remains robust regardless of the coverage definition. Notably, the quantile-defined intervals produced by Cal-Gibbs are consistently wider than their content-defined counterparts. For example, in the $\mathcal{N}(10,3)$ case with $P=0.90$, the average length increases from $11.950$ (content) to $13.536$ (quantile). This increased width is a necessary consequence of the more stringent quantile-defined criterion, which requires simultaneous coverage of two specific tail probabilities rather than a single aggregate mass. This adaptability is driven by the calibration of the learning rate $\eta$, which takes on different values depending on whether the objective is targeting specific population quantiles or overall content. By adjusting the posterior spread through $\eta$, the method ensures reliable and efficient tolerance limits tailored to the inferential requirements at hand.

\section{Applications}
We illustrate the calibrated Gibbs posterior methodology using three real-world datasets from ecology, biopharmaceutical drug development, and environmental health. These examples span a wide range of sample sizes, from the large-scale longleaf pine dataset to the smaller air lead level dataset. In each case, we implement the learning rate calibration procedure and construct tolerance intervals targeting the desired frequentist coverage. Together, these applications demonstrate how the calibration algorithm performs across diverse data structures and sample sizes.

\subsection{Longleaf Pines Data}
The first application involves data from a census of longleaf pines (\emph{Pinus palustris}) located in the Wade Tract, an old-growth forest in Thomas County, Georgia. Originally surveyed by \cite{platt1988population}. The dataset represents a marked point pattern where each tree’s coordinate is associated with its diameter at breast height (dbh) as a measure of size. The study area covers a 200 $\times$ 200 meter region containing 584 individual trees. While the complete census includes all specimens with a diameter of at least 2 cm, many ecological analyses focus on the ``adult" population, typically defined by a size threshold of 30 cm or greater. This dataset is characterized by spatial inhomogeneity and has been used previously in the context of non-parametric tolerance intervals, notably by \cite{frey:2010}. The data are publicly available and can be accessed through the \texttt{spatstat} package \citep{baddeley2015spatial} in \texttt{R}. Our analysis adopts this dataset as a benchmark for constructing two-sided tolerance intervals via the calibrated Gibbs posterior, focusing on the distributional properties of tree diameters without assuming a specific parametric form.

The results of the analysis demonstrate close alignment with established nonparametric methods while offering the added flexibility of Bayesian-inspired uncertainty quantification. The Wilks method yields a two-sided $(P = 0.50, 1 - \alpha = 0.90)$ tolerance interval of $(8.1, 42.9)$, and the YM method provides a slightly tighter bound of $(8.1, 42.7)$. Our proposed calibrated Gibbs method produces a content-based interval of $(8.46, 42.68)$, which is notably more efficient (narrower) than both nonparametric benchmarks while still maintaining the required frequentist guarantees. Furthermore, the calibrated Gibbs method allows for the construction of a quantile-specific interval, which targets the central 50\% of the population by bounding the 0.25 and 0.75 quantiles. This results in an interval of $(7.05, 43.92)$. As noted in our simulation studies, the $\eta$ calibration adapts to the specific inferential goal, resulting in a slightly wider interval for the quantile-defined case to ensure that both individual tail requirements are satisfied with the nominal 90\% confidence. This distinction highlights the utility of the calibrated Gibbs framework in ecological monitoring, where management objectives often rely on specific population thresholds (such as minimum size requirements for adult classification) rather than aggregate mass alone.

\subsection{Relative Potency Data}
The second application centers on relative potency measurements from biopharmaceutical manufacturing, a dataset introduced by \cite{lewis:2022} and subsequently used by \cite{choi:2025} to evaluate nonparametric tolerance intervals. In this context, relative potency compares the biological activity of a manufactured batch to a reference standard. Maintaining potency within a strict 90\% to 110\% window is a critical quality requirement; deviations from this range can lead to compromised efficacy or increased safety risks. The dataset consists of 25 observations spanning five distinct manufacturing campaigns, as shown in Table \ref{tab:potency_data}.

\begin{table}[ht]
\centering
\begin{tabular}{ccccc}
\hline
Campaign 1 & Campaign 2 & Campaign 3 & Campaign 4 & Campaign 5 \\ \hline
95.661     & 98.830     & 96.665     & 98.198     & 95.922     \\
102.259    & 94.887     & 106.234    & 98.186     & 102.956    \\
103.135    & 103.362    & 103.735    & 107.872    & 101.596    \\
99.827     & 94.117     & 104.317    & 99.987     & 96.806     \\
           &            & 101.807    & 103.051    & 107.041    \\
           &            &            & 106.445    & 92.589     \\ \hline
\end{tabular}
\caption{Relative potency data (as a percentage) for 25 samples across five manufacturing campaigns.}\label{tab:potency_data}
\end{table}

Following the approach of \cite{choi:2025}, we treat these observations as a representative sample of the overall manufacturing process to establish a population-level tolerance interval, rather than modeling individual batch effects. By using the calibrated Gibbs posterior, we construct a two-sided tolerance interval to establish whether the relative potency measurements meet the defined manufacturing specifications.

Following the specifications used by \cite{choi:2025}, we consider the construction of a two-sided $(P = 0.95, 1 - \alpha = 0.95)$ tolerance interval. This setting poses a substantial challenge for traditional nonparametric procedures. In particular, the Wilks method requires the sample size to satisfy the two-sided feasibility condition $(n - 1)P^n - nP^{n-1} + 1 \ge 1 - \alpha$, which implies a minimum of $n = 93$ observations for a valid interval based on sample extremes. Because the available dataset contains only $n = 25$ measurements, the Wilks method is mathematically inapplicable. Nevertheless, benchmarking remains possible using the YM method, which employs interpolated order statistics to construct tolerance limits in small-sample settings.

The practical implications of the chosen methodology are important in this regulatory context. The YM method yields a tolerance interval of (88.66, 110.01), marginally exceeding the 90\%--110\% specification limits.. In a strict quality-control assessment, an  interval of this form could raise concerns regarding process capability and might prompt additional investigation or monitoring of the manufacturing process.

The calibrated Gibbs posterior produces alternative tolerance intervals based on two calibration targets. The content-based calibrated Gibbs interval is $(91.27, 109.21)$, while the quantile-based calibrated Gibbs interval is $(88.87, 111.98)$. Both intervals achieve the nominal $(P = 0.95, 1-\alpha = 0.95)$ coverage through the learning-rate calibration procedure described in Section~\ref{sec:calibration}. Relative to the YM interval, the calibrated Gibbs intervals distribute uncertainty differently in small samples. The content-based interval is more concentrated around the central portion of the data, whereas the quantile-based interval is wider due to the direct calibration of the bounding quantiles. These differences illustrate how alternative tolerance constructions can produce different interval widths while maintaining nominal coverage. From a practical perspective, this example highlights the sensitivity of tolerance interval conclusions to the chosen construction method when the sample size is limited. Order-statistic procedures rely heavily on sample extremes, whereas the calibrated Gibbs approach incorporates information from the full sample through the loss-based posterior formulation.

\subsection{Air Lead Levels}
The final application involves environmental monitoring data collected by the National Institute for Occupational Safety and Health (NIOSH). The dataset, which records air lead concentrations ($\mu g/m^3$) at $n=15$ different laboratory locations, was originally discussed by \cite{krishnamoorthy2009statistical} and later by \cite{young:2014}. Environmental data of this type are often expensive to obtain, resulting in the small sample sizes seen here. Previous analyses have typically relied on a log-transformation to satisfy normality assumptions before constructing tolerance limits. In contrast, we use this dataset to demonstrate the performance of the calibrated Gibbs posterior in a small-sample setting, where we construct a one-sided upper tolerance limit directly. The 15 observations are provided in Table \ref{tab:air_lead}.

\begin{table}[ht]
\centering
\begin{tabular}{cccccccc}
\hline
200 & 120 & 15 & 7 & 86 & 48 & 61 & 380\\
 80 & 29 & 1000 & 350 & 1400 & 110 & 37 & \\ \hline
\end{tabular}
\caption{Air lead levels ($\mu g/m^3$) collected at $n=15$ locations by NIOSH.}\label{tab:air_lead}
\end{table}

Following the setup in \cite{young:2014}, the objective is to construct a one-sided upper $(P = 0.75, 1 - \alpha = 0.85)$ tolerance limit for air lead concentrations. This configuration is notable because the minimum sample size required for the Wilks method is relatively small, $n \ge \ln(\alpha)/\ln(P) \approx 6.6$. With $n = 15$ observations available, the nonparametric feasibility conditions are easily satisfied, though the resulting limits may still be conservative. Using the Wilks method, the upper tolerance limit corresponds to the 14th order statistic and yields an estimate of $1000.00$. The YM method, which employs interpolated order statistics to mitigate discreteness, produces a lower but still conservative limit of $722.35$.

This dataset provides a challenging stress test for the calibrated Gibbs posterior due to its heavy-tailed and highly skewed structure. When applying the standard Robbins-Monro stochastic approximation procedure to calibrate the learning rate $\eta$, the algorithm fails to stabilize at a positive value, with $\eta$ repeatedly drifting toward zero. This behavior indicates that, under the asymmetric Laplace working likelihood, the posterior uncertainty cannot be adequately adjusted to achieve the desired frequentist coverage through local stochastic updates alone. To address this issue, we performed a fine-grained grid search over $\eta \in (0, 1]$, identifying an optimal learning rate of $\eta = 0.0034$. The resulting value is exceptionally small, reflecting the need to substantially inflate posterior uncertainty in order to maintain the nominal 85\% confidence level. Using this calibrated learning rate, the calibrated Gibbs posterior yields an upper tolerance limit of $436.01$.

This case study highlights an important practical consideration: while the calibrated Gibbs posterior can produce substantially more efficient tolerance limits than nonparametric benchmarks, the calibration of the learning rate is not always a trivial optimization task. In sparse or highly skewed settings, stochastic approximation methods may fail to converge, necessitating alternative calibration strategies such as the grid search employed here. Nonetheless, once a suitable learning rate is identified, the calibrated Gibbs posterior continues to deliver valid frequentist coverage with noticeably tighter bounds than traditional order-statistic-based approaches.

\section{Discussion}
\label{sec:discussion}
The results of this work demonstrate that the calibrated Gibbs posterior provides a robust and flexible framework for constructing nonparametric tolerance intervals. By reframing the tolerance interval problem as an inferential task on population quantiles, the methodology bridges the gap between Bayesian uncertainty quantification and frequentist coverage guarantees, offering a unified approach to tolerance interval estimation.

A key advantage of this framework is its ability to overcome the inherent rigidity of traditional nonparametric methods. Classical procedures, such as Wilks-type intervals, often require large sample sizes and lack adaptability to alternative definitions of coverage. In contrast, the Gibbs posterior leverages the full data distribution through the check loss function, rather than relying solely on extreme order statistics, enabling precise and valid tolerance bounds even in small-sample settings.

Another important contribution is the distinction between content-defined and quantile-defined coverage for two-sided intervals. By calibrating the learning rate $\eta$, the framework can be tuned to satisfy either criterion. Quantile-defined intervals are naturally wider due to their coverage of specific population tails, but they provide a level of inferential precision that mass-based intervals cannot match. This flexibility allows practitioners to align interval construction with the specific inferential needs of their applications, whether the focus is on aggregate population coverage or on tail-specific thresholds.

Several avenues for future research arise naturally from this work. Extensions to regression settings could enable tolerance bands that vary with covariates, while multivariate loss functions could facilitate nonparametric tolerance regions for higher-dimensional outcomes. Additional developments might explore alternative loss functions for increased robustness, approximate sampling methods for computational efficiency, or integration with hierarchical models for structured data. Such extensions would broaden the applicability of calibrated Gibbs inference while preserving its core strengths: nonparametric reliability, flexible coverage definitions, and coherent uncertainty quantification.
\section*{Supplementary material}
The Supplementary Material includes: (i) Proofs of  Theorem
\ref{thm:onesided-main} and Theorem \ref{thm:twosided-main}; (ii) mathematical
details for all comparative methods in Section \ref{subsec:comparative}; (iii)
computational implementation and tuning details; and (iv) expanded simulation
tables and figures for practical considerations (Section
\ref{subsec:considerations}). (PDF file)

\section*{Data Availability Statement}
The longleaf pine dataset analyzed in this paper is publicly available through the \texttt{spatstat} R package. The air lead levels data and the relative potency data are reported in Tables \ref{tab:air_lead} and \ref{tab:potency_data}, respectively. All data used in the simulation studies are generated from known distributions. 

\begin{singlespace}
\bibliographystyle{apalike}
\bibliography{references}
\end{singlespace}

\newpage

\setcounter{section}{0}
\section*{Supplementary Material}

\section{Asymptotic Validity of Calibrated Gibbs Tolerance Intervals}

Let $Y_1,\dots,Y_n$ be i.i.d. from an unknown distribution $F$ with density $f$. For
$\tau\in(0,1)$, define empirical risk as
$ R_{n,\tau}(q)=\sum_{i=1}^n \rho_\tau(Y_i-q).$ We make the following assumptions:
\begin{itemize}
\item $F$ has a density $f$ continuous at $Q_\tau$ with $f(Q_\tau)>0$.
\item The prior $\pi_0(q)$ is continuous and strictly positive in a neighborhood of $Q_\tau$.
\item Calibration restricts $\eta$ to a compact interval
\[
[\underline{\eta},\overline{\eta}] \subset (0,\infty)
\]
and returns $\hat\eta_n \in [\underline{\eta},\overline{\eta}]$.
\item The estimated coverage function $\widehat C_n(\eta)$ used in calibration satisfies
\[
\sup_{\eta\in[\underline{\eta},\overline{\eta}]}
|\widehat C_n(\eta)-C_n(\eta)|
\overset{p}{\longrightarrow}0.
\]
\end{itemize}

\subsection{One-Sided Interval}

Fix $\tau=P$ and let $Q_\tau$ be the population $\tau$-quantile. Let $\hat Q_\tau$ be the
empirical minimizer of $R_{n,\tau}(q)$ and define the Gibbs posterior $
\pi_\eta(q\mid Y_{1:n})\propto \exp\{-\eta R_{n,\tau}(q)\}\pi_0(q), $ with learning rate $\eta>0$.
Let $U_n(\eta)$ denote the $(1-\alpha)$ posterior quantile of $\pi_\eta(\cdot\mid Y_{1:n})$.

\paragraph{Proof of Theorem 1.}

Using Knight's identity and a local expansion around $Q_\tau$, for bounded $t$,
\[
R_{n,\tau}\!\left(Q_\tau+\frac{t}{\sqrt n}\right)-R_{n,\tau}(Q_\tau)
=-\frac{t}{\sqrt n}\sum_{i=1}^n\psi_\tau(Y_i-Q_\tau)+\frac{f(Q_\tau)}{2}t^2+o_p(1),
\]
where $\psi_\tau(u)=\tau-\1(u<0)$. Re-centering at $\hat Q_\tau=Q_\tau+O_p(n^{-1/2})$ gives
\[
R_{n,\tau}\!\left(\hat Q_\tau+\frac{t}{\sqrt n}\right)-R_{n,\tau}(\hat Q_\tau)
=\frac{f(Q_\tau)}{2}t^2+o_p(1).
\]
Hence, by Laplace approximation,
\[
\sqrt n(q-\hat Q_\tau)\mid Y_{1:n}\Rightarrow N\!\left(0,\frac{1}{\eta f(Q_\tau)}\right),
\]
so
\[
U_n(\eta)=\hat Q_\tau+z_{1-\alpha}\sqrt{\frac{1}{\eta n f(Q_\tau)}}+o_p(n^{-1/2}).
\]
The quantile CLT gives
\[
\sqrt n(\hat Q_\tau-Q_\tau)\Rightarrow N\!\left(0,\frac{\tau(1-\tau)}{f(Q_\tau)^2}\right).
\]
Therefore
\[
C_n(\eta)=P_F\{Q_\tau\le U_n(\eta)\}\to
C(\eta)=\Phi\!\left(z_{1-\alpha}\sqrt{\frac{f(Q_\tau)}{\eta\tau(1-\tau)}}\right).
\]
Since $C(\eta)$ is continuous and strictly decreasing, $C(\eta)=1-\alpha$ has the unique solution $\eta^\star=f(Q_\tau)/\{\tau(1-\tau)\}$. Uniform consistency of $\widehat C_n$ implies $\hat\eta_n\overset{p}{\to}\eta^\star$, and substituting into $C_n$ yields the desired coverage. \hfill$\square$

\subsection{Two-Sided Interval}

Fix $0<\tau_L<\tau_U<1$, set $Q_L=Q_{\tau_L}$, $Q_U=Q_{\tau_U}$, and $P=\tau_U-\tau_L$.
Define the joint empirical risk
\[
R_n(q_L,q_U)=\sum_{i=1}^n\big\{\rho_{\tau_L}(Y_i-q_L)+\rho_{\tau_U}(Y_i-q_U)\big\},
\]
with minimizer $(\hat Q_L,\hat Q_U)$.
The joint Gibbs posterior is
\[
\pi_\eta(q_L,q_U\mid Y_{1:n})\propto \exp\{-\eta R_n(q_L,q_U)\}\pi_0(q_L,q_U).
\]
Let $\bar\mu=E[(q_L+q_U)/2\mid Y_{1:n}]$ and define, for a posterior draw,
\[
U^\star=\max\{q_U,\,2\bar\mu-q_L\}.
\]
Let $U_n(\eta)$ be the $(1-\alpha)$ posterior quantile of $U^\star$, and set
$L_n(\eta)=2\bar\mu-U_n(\eta)$.

\paragraph{Proof of Theorem 2.}

Applying Knight's identity jointly yields, uniformly for bounded $(t_L,t_U)$,
$$\begin{aligned}
R_n \left( Q_L + \frac{t_L}{\sqrt n}, Q_U + \frac{t_U}{\sqrt n} \right) &- R_n(Q_L, Q_U) \\
&= -\frac{t_L}{\sqrt n} \sum_{i=1}^n \psi_{\tau_L}(Y_i - Q_L) - \frac{t_U}{\sqrt n} \sum_{i=1}^n \psi_{\tau_U}(Y_i - Q_U) \\
&\quad + \frac{f(Q_L)}{2}t_L^2 + \frac{f(Q_U)}{2}t_U^2 + o_p(1)
\end{aligned}$$
Re-centering at $(\hat Q_L,\hat Q_U)$ gives a local quadratic form, so Laplace approximation implies
\[
\sqrt n\big((q_L,q_U)-(\hat Q_L,\hat Q_U)\big)\mid Y_{1:n}
\Rightarrow N_2\!\left(0,\operatorname{diag}\{(\eta f(Q_L))^{-1},(\eta f(Q_U))^{-1}\}\right).
\]
Therefore
\[
U_n(\eta)=\hat Q_U+\frac{m_{1-\alpha}(\eta)}{\sqrt n}+o_p(n^{-1/2}),
\qquad
L_n(\eta)=\hat Q_L-\frac{m_{1-\alpha}(\eta)}{\sqrt n}+o_p(n^{-1/2}),
\]
where $m_{1-\alpha}(\eta)$ is the $(1-\alpha)$ quantile of $M_\eta=\max\{Z_U,-Z_L\}$ with
$Z_U\sim N(0,1/(\eta f(Q_U)))$, $Z_L\sim N(0,1/(\eta f(Q_L)))$. Since $M_\eta$ scales as
$\eta^{-1/2}$, $m_{1-\alpha}(\eta)$ is continuous and strictly decreasing in $\eta$.

Next, from the joint quantile CLT,
\[
\sqrt n\begin{pmatrix}\hat Q_L-Q_L\\ \hat Q_U-Q_U\end{pmatrix}
\Rightarrow N_2(0,\Sigma_Q),
\]
which implies limit coverage functions for both calibration criteria:
\[
C_n^{(Q)}(\eta)\to C^{(Q)}(\eta),
\qquad
C_n^{(C)}(\eta)\to C^{(C)}(\eta)
=\Phi\!\left(\frac{(f(Q_U)+f(Q_L))m_{1-\alpha}(\eta)}{\sqrt{P(1-P)}}\right).
\]
Both are continuous and strictly decreasing in $\eta$, so each has a unique root at level
$1-\alpha$ within the calibration bracket. Uniform consistency of estimated coverage then gives
$\hat\eta_n\overset{p}{\to}\eta_Q^\star$ or $\eta_C^\star$, respectively, and plug-in yields
asymptotic nominal coverage under each success criterion. \hfill$\square$

Although the unconstrained Gibbs posterior on $(q_L,q_U)$ may assign nonzero finite-sample mass to the region $q_L>q_U$,
this is asymptotically negligible: since $Q_L<Q_U$ and the posterior concentrates at $n^{-1/2}$ rate around $(Q_L,Q_U)$,
we have $\Pi_\eta(q_L>q_U\mid Y)\to 0$ in probability.

\section{Sensitivity of the GPC Algorithm}
The practical performance of the Generalized Posterior Calibration (GPC) procedure depends on several implementation choices, including the initial value of $\eta$, the step-size parameters $(c,\gamma)$, the number of bootstrap replicates $B$, and the maximum number of iterations. In finite samples, the estimated coverage function $\hat{C}(\eta)$ can exhibit substantial Monte Carlo variability, particularly for heavy-tailed or highly skewed distributions. As a result, the Robbins--Monro updates may converge slowly, oscillate, or settle at suboptimal values if these tuning parameters are poorly chosen.

In routine applications, we therefore recommend an initial pilot phase in which the calibration procedure is run under a range of starting values and algorithmic settings. These pilot runs can be used to identify initial values of $\eta$ that produce reasonably stable bootstrap coverage estimates and to assess whether the stochastic approximation updates are progressing in a consistent direction. In both one- and two-sided cases, the pilot phase can be further refined or supplemented by initializing $\eta_0$ using a plug-in estimate of the theoretical optimum $\eta^*$. In the one-sided case, this is $\eta^* = f(Q_{\tau}) / \{ \tau(1-\tau) \}$ ; in the two-sided case, it involves the sum of densities $f(Q_L) + f(Q_U)$. While these densities are unknown, they can be estimated via kernel density estimation on the original sample. Monitoring the trajectory of $\eta_s$ and the associated estimated coverage values across iterations provides a simple diagnostic for detecting overly aggressive or insufficient updating. For example, persistent under-coverage despite a decreasing $\eta$ may signal the need for a more expansive step-size or a higher initial learning rate. In our experience, moderate step sizes and gradual decay rates tend to promote stable convergence, while overly aggressive choices of $(c,\gamma)$ may lead to persistent undercoverage or premature stabilization.

Our empirical investigations suggest that using approximately 10--20 Robbins--Monro iterations and on the order of $B=200$ bootstrap replicates provides a reasonable balance between computational cost and calibration accuracy in many settings, consistent with previous recommendations \citep{martin:2022}. When combined with a carefully chosen initial value of $\eta$, these settings typically yield calibrated tolerance intervals with reliable frequentist performance. However, we emphasize that GPC calibration is not fully automatic, and its effectiveness depends on thoughtful preliminary exploration. 

The supplementary simulation results illustrate how different tuning regimes affect convergence behavior, coverage accuracy, and interval length, and demonstrate the value of pilot-based tuning in practical applications. The tables and figure presented here provide diagnostic results supporting the discussion of practical considerations in Section 4.2. Although the generalized posterior calibration (GPC) procedure is theoretically justified \citep{martin:2022}, finite-sample performance can vary substantially depending on the initial learning rate, step-size parameters, and iteration scheme. In particular, the coverage estimate $\hat{C}(\eta_s)$ is subject to Monte Carlo variability arising from both bootstrap sampling and posterior draws. To mitigate this, we recommend conducting initial pilot runs. These runs serve a dual purpose: they allow the researcher to identify a well-informed starting value for $\eta_0$ and provide a diagnostic look at how the Robbins--Monro updates are behaving under specific hyperparameters. These results quantify this variability and illustrate the consequences of different calibration strategies, showing how choices in the stochastic approximation procedure affect convergence speed, stability, and interval length. 

For these diagnostics, we focus on the one-sided upper tolerance interval for the Normal(0,1) and Pareto(1,2) distributions. The same calibration behavior arises for two-sided tolerance intervals, but we restrict attention to the one-sided case to illustrate the dynamics of the learning-rate updates. For each distribution, we ran $R=1{,}000$ independent calibration experiments using three Robbins--Monro strategies \citep{robbins:1951}, referred to as \textit{Conservative}, \textit{Moderate}, and \textit{Aggressive}. The \textit{Conservative} regime is characterized by a small step-size and limited iterations; the \textit{Moderate} regime represents a balanced configuration intended for routine use; and the \textit{Aggressive} regime employs high-gain updates to explore the boundaries of the algorithm's numerical stability. Each strategy differs in its choice of step-size hyperparameters $(\eta_0, c, \gamma)$ and maximum iterations (max\_iter), which control the starting value, magnitude, and decay of updates in $\eta_s$. Table sets 1–4 summarize the calibrated learning rates $\hat{\eta}$, interval lengths, and coverage achieved across these regimes, while Figure \ref{fig:eta_traj} shows representative $\eta_s$ trajectories. Across both distributions, the Moderate regime consistently provides the most reliable trade-off between stability, coverage, and interval length, suggesting it as a practical default. While these diagnostics focus on two representative distributions, we observed similar qualitative behavior in additional experiments (not shown).

\subsection*{Normal Distribution}

The Normal(0,1) regimes were configured as follows:
\begin{itemize}
    \item Aggressive: $\eta_0 = 1.5$, $c = 1.75$, $\gamma = 0.5$, $\text{max\_iter} = 10$
    \item Moderate: $\eta_0 = 2.8$, $c = 0.5$, $\gamma = 0.75$, $\text{max\_iter} = 25$
    \item Conservative: $\eta_0 = 3.3$, $c = 0.2$, $\gamma = 0.9$, $\text{max\_iter} = 6$
\end{itemize}

These values were chosen based on preliminary pilot experiments to represent low-, medium-, and high-gain update regimes.

\begin{table}[ht]
\centering
\begin{tabular}{lcccc}
\hline
Strategy & Calibrated Cov. & Actual Coverage & Calibrated $\hat{\eta}$ & Interval Length \\ 
\hline
Aggressive   & 0.953 (0.030) & 0.939 (0.239) & 1.987 (0.224) & 1.843 (0.389) \\ 
Moderate     & 0.882 (0.065) & 0.891 (0.312) & 2.747 (0.173) & 1.723 (0.366) \\ 
Conservative & 0.834 (0.094) & 0.862 (0.345) & 3.276 (0.034) & 1.670 (0.357) \\ 
\hline
\end{tabular}
\caption{GPC performance for the Normal distribution ($n=22, P=0.90, 1-\alpha=0.90$). Mean (SD) reported over $R=1{,}000$ independent simulations per strategy.}
\label{tab:regime_summary_normal}
\end{table}

Table \ref{tab:regime_summary_normal} shows the average performance of each regime. Differences between calibrated and actual coverage reflect finite-sample error in estimating the coverage function and indicate residual Monte Carlo noise in the calibration step. The Moderate strategy achieves a good balance between coverage and interval length. The Conservative regime yields slightly lower coverage but comparable interval lengths, while the Aggressive regime attains slightly higher coverage with modestly longer intervals. Table \ref{tab:overall_summary_normal} pools all runs across strategies. Most runs have calibrated learning rates near the median (~2.8) and interval lengths near 1.7, but the full range of $\hat{\eta}$ (1.027–3.336) indicates that some runs produce more extreme updates. Similarly, interval lengths vary from 0.611 to 3.155, suggesting that while typical performance is stable, the stochastic nature of bootstrap and posterior sampling can occasionally produce outliers. 

\subsection*{Pareto Distribution}
We now turn to the more challenging heavy-tailed case. The Pareto(1,2) regimes were configured as follows:
\begin{itemize}
\item Aggressive: $\eta_0 = 5.0$, $c = 1.75$, $\gamma = 0.5$, $\text{max\_iter} = 10$
\item Moderate: $\eta_0 = 1.15$, $c = 0.5$, $\gamma = 0.75$, $\text{max\_iter} = 25$
\item Conservative: $\eta_0 = 4.0$, $c = 0.2$, $\gamma = 0.9$, $\text{max\_iter} = 6$
\end{itemize}

\begin{table}[ht]
\centering
\begin{tabular}{lcccccc}
\hline
Metric & Min. & 1st Qu. & Median & Mean & 3rd Qu. & Max. \\ 
\hline
Learning Rate ($\hat{\eta}$) & 1.027 & 2.171 & 2.789 & 2.670 & 3.254 & 3.336 \\ 
Interval Length             & 0.611 & 1.474 & 1.721 & 1.745 & 2.004 & 3.155 \\ 
\hline
\end{tabular}
\caption{Summary statistics for calibrated parameters across all pooled Normal diagnostic runs ($R=3{,}000$).}
\label{tab:overall_summary_normal}
\end{table}

Table \ref{tab:regime_summary_pareto} shows regime-level performance for the heavy-tailed Pareto distribution. Compared with the Normal case, the variability in calibrated learning rates and interval lengths is markedly larger. The Moderate regime again provides a balanced configuration, while the Aggressive and Conservative regimes produce extreme outcomes in some runs due to either high-gain updates or very slow adjustment. Table \ref{tab:overall_summary_pareto} pools all runs and illustrates the large spread: $\hat{\eta}$ ranges from 0.037 to 5.467 and interval lengths from 1.743 up to 52.801, reflecting the challenges of heavy-tailed data and the stochastic nature of bootstrap-based calibration. Most runs cluster near the median values (2.629 for $\hat{\eta}$, 4.048 for interval length), but extreme draws are not uncommon.

\begin{table}[ht]
\centering
\begin{tabular}{lcccc}
\hline
Strategy & Calibrated Cov. & Actual Coverage & Calibrated $\hat{\eta}$ & Interval Length \\ 
\hline
Aggressive   & 0.705 (0.071) & 0.742 (0.438) & 2.515 (1.096) & 4.650 (3.269) \\ 
Moderate     & 0.892 (0.075) & 0.900 (0.300) & 1.076 (0.272) & 5.141 (3.027) \\ 
Conservative & 0.542 (0.193) & 0.713 (0.453) & 3.870 (0.070) & 4.136 (2.024) \\ 
\hline
\end{tabular}
\caption{GPC performance for the Pareto distribution ($n=22, P=0.90, 1-\alpha=0.90$). Mean (SD) reported over $R=1{,}000$ independent simulations per strategy.}
\label{tab:regime_summary_pareto}
\end{table}

\begin{table}[ht]
\centering
\begin{tabular}{lcccccc}
\hline
Metric & Min. & 1st Qu. & Median & Mean & 3rd Qu. & Max. \\ 
\hline
Learning Rate ($\hat{\eta}$) & 0.037 & 1.201 & 2.629 & 2.487 & 3.839 & 5.467 \\ 
Interval Length             & 1.743 & 3.269 & 4.048 & 4.643 & 5.136 & 52.801 \\ 
\hline
\end{tabular}
\caption{Summary statistics for calibrated parameters across all pooled Pareto diagnostic runs ($R=3{,}000$).}
\label{tab:overall_summary_pareto}
\end{table}

Figure \ref{fig:eta_traj} visualizes representative trajectories of $\eta_s$ for each calibration regime. The Moderate regime demonstrates stable convergence toward the target learning rate, while the Aggressive regime exhibits volatility and occasional overshooting. The Conservative regime updates slowly, showing inertia in early iterations. These trajectories provide a visual diagnostic for assessing convergence and highlight the importance of tuning step sizes and initial values to achieve robust calibration.

\begin{figure}[ht]
\centering
\includegraphics[width=\textwidth]{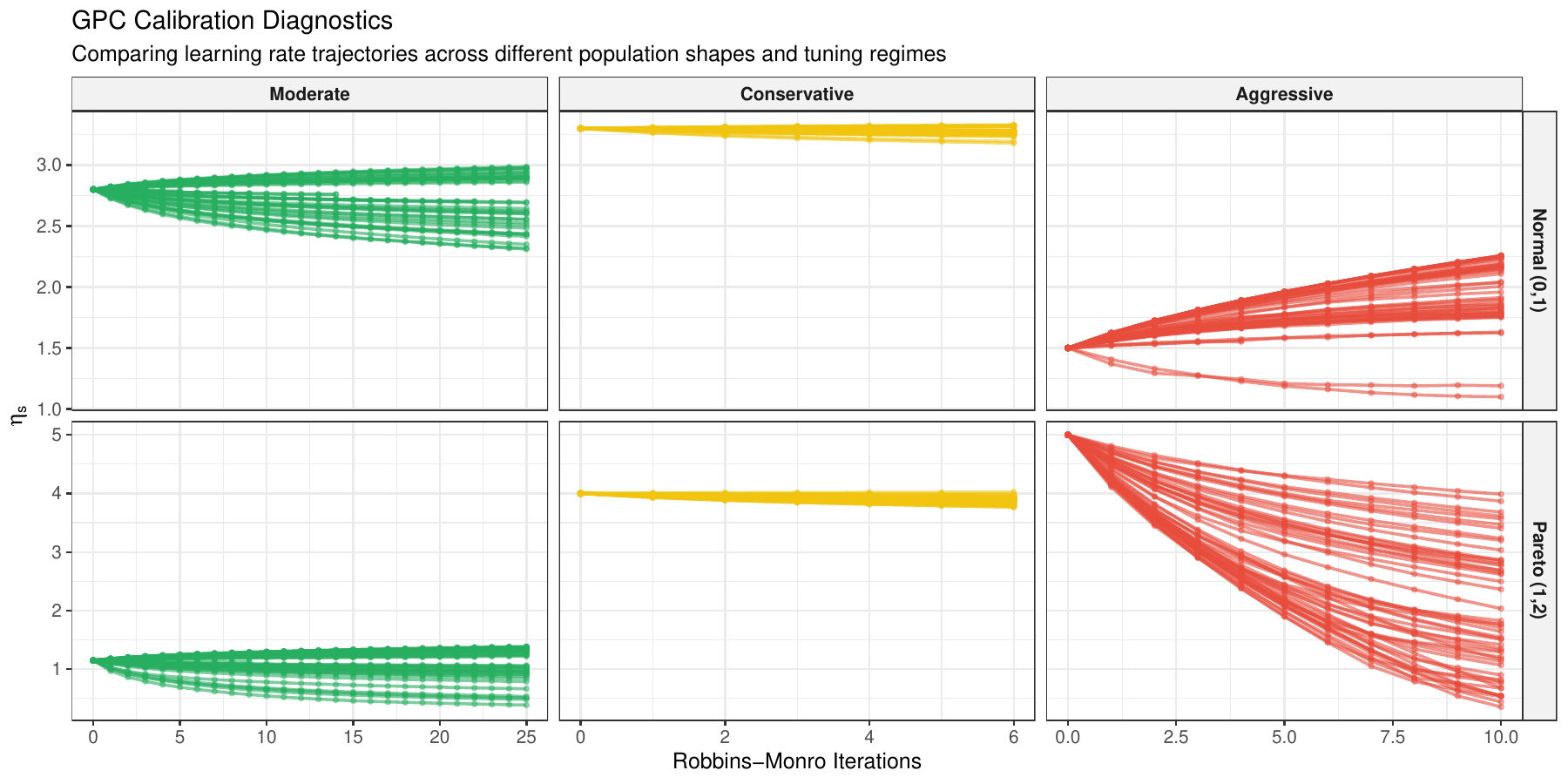}
\caption{Trajectories of $\eta_s$ across Normal and Pareto distributions. The grid contrasts the stable convergence of the Moderate regime against the volatility of the Aggressive regime and the slow inertia of the Conservative regime.}
\label{fig:eta_traj}
\end{figure}

Across both distributions, the results illustrate that Moderate step sizes provide a favorable balance between convergence speed and stability, while Aggressive updates can oscillate and Conservative updates may converge too slowly. Heavy-tailed distributions such as the Pareto show increased variability in $\hat{\eta}$ and interval lengths, highlighting the practical importance of careful pilot tuning. In practice, we recommend beginning with a Moderate configuration and adjusting toward Conservative or Aggressive regimes only when trajectory diagnostics indicate persistent bias or stagnation. Overall, these results support the recommendations in Section 4.2. Practitioners can fine-tune the calibration for their specific dataset by using pilot runs to observe trajectory patterns, which allows for adjusting $c$ and $\gamma$ to reduce volatility or increasing the number of Robbins-Monro iterations to overcome inertia. With thoughtful selection of initial values and step-size parameters, GPC calibration can reliably produce tolerance intervals with near-nominal frequentist coverage. These findings reinforce the discussion in Section 4.2 regarding the importance of trajectory diagnostics and pilot calibration.

\section{Comparative Methods}
The proposed calibrated Gibbs (Cal-Gibbs) posterior is evaluated against several established nonparametric and Bayesian benchmarks.
\subsubsection*{Wilks Method}
The classical nonparametric approach, introduced by \cite{wilks:1941}, defines tolerance limits using order statistics $Y_{(1)} \le Y_{(2)} \le \dots \le Y_{(n)}$. The goal is to select the appropriate order statistics to serve as limits such that the probability of the population coverage exceeding $P$ is at least $1-\alpha$. For a $(P, 1-\alpha)$ upper one-sided tolerance limit $U = Y_{(k)}$, we seek the smallest integer $k$ such that
\begin{align}
    \text{Pr}(F(Y_{(k)}) \ge P) = \sum_{j=0}^{k-1} \binom{n}{j} P^j (1-P)^{n-j} \ge 1-\alpha,
\end{align}
whereas for a $(P, 1-\alpha)$ lower one-sided tolerance limit $L = Y_{(k)}$, we seek the largest integer $k$ such that
\begin{align}
  \text{Pr}(F(Y_{(k)}) \leq 1-P) = \sum_{j=k}^{n} \binom{n}{j} (1-P)^j P^{n-j} \ge 1-\alpha. 
\end{align}
In the two-sided case $(L, U) = (Y_{(r)}, Y_{(s)})$, the coverage condition is determined by the number of observations contained within the limits, $m = s - r$, where we seek the smallest $m$ satisfying
\begin{align}
    \text{Pr}(F(Y_{(s)}) - F(Y_{(r)}) \ge P) = \sum_{j=0}^{s-r-1} \binom{n}{j} P^j (1-P)^{n-j} \ge 1-\alpha.
\end{align}
A major limitation of the Wilks method is its conservative nature arising from discreteness. Because the indices $k$, $r$, and $s$ must be integers, the achieved coverage probability often exceeds the nominal level $1-\alpha$, particularly for small to moderate sample sizes. As a result, the resulting tolerance limits are frequently more extreme than necessary, leading to wider intervals and reduced efficiency in practical applications.
\subsubsection*{Interpolated and Extrapolated Order Statistics (YM)}
To mitigate the discreteness and resulting conservatism of Wilks-type tolerance limits, \cite{young:2014} propose the use of interpolated order statistics. In the one-sided setting, the integer index $k$ appearing in the Wilks coverage equation is replaced by a fractional index $v \in \mathbb{R}$ obtained by solving
$\text{Pr}\!\left(F(Y_{(v)}) \ge P\right) = 1-\alpha$
as if $k$ were a continuous variable. The resulting tolerance limit is then defined via linear interpolation between adjacent order statistics,
\begin{align}
    Y_{(v)} = (\lfloor v \rfloor + 1 - v) Y_{(\lfloor v \rfloor)} + (v - \lfloor v \rfloor) Y_{(\lfloor v \rfloor + 1)}.
\end{align}
For two-sided tolerance intervals, the authors note that a direct extension of fractional order statistics can lead to unstable behavior. Instead, they propose a coverage-correcting interpolation at the interval endpoints. Let $k$ denote the smallest integer satisfying the Wilks coverage condition and define an interpolation weight
\[
\lambda = \frac{(1-\alpha) - \Pr(B \le k-2)}{\Pr(B = k-1)},
\]
where $B \sim \mathrm{Binomial}(n,P)$. The final tolerance interval is obtained by interpolating either the lower or upper endpoint using weight $\lambda$, with the shorter resulting interval selected.

When the required index exceeds the available sample size, which commonly occurs for small samples and high content levels, \cite{young:2014} further propose linear extrapolation beyond the extreme order statistics. Although extrapolation remains heuristic and may still under-cover in very small samples, their simulations demonstrate improved performance relative to using the sample extrema alone.

\subsubsection*{Bayesian Quantile Regression (BQR-AL)}
The check loss function $\rho_\tau(\cdot)$ used in the construction of our Gibbs posterior is closely related to the Asymmetric Laplace (AL) likelihood commonly employed in Bayesian Quantile Regression (BQR-AL) \citep{yu:2001}. Under the BQR-AL framework, the data are assumed to follow an AL distribution:
\begin{align}
f(y \mid Q_\tau, \sigma) = \frac{\tau(1-\tau)}{\sigma} \exp \left\{ -\frac{\rho_\tau(y - Q_\tau)}{\sigma} \right\},
\end{align}
where $\tau\in(0,1)$, $Q_\tau$ is the quantile of interest, and $\sigma$ is a scale parameter. In this traditional Bayesian setting, $\sigma$ is treated as a model parameter with its own prior distribution (e.g., an inverse-gamma prior), and inference is performed on the joint posterior of $(Q_\tau, \sigma)$. There is a direct mathematical connection between this likelihood and the Gibbs posterior: setting $\sigma = 1/\eta$ renders the AL likelihood proportional to the Gibbs kernel.

The Asymmetric Laplace distribution is typically employed as a working likelihood because its kernel coincides with the check loss function, providing a convenient Bayesian mechanism for quantile inference. However, since the data-generating process $F$ is unknown and unlikely to follow an AL distribution, the model is inherently misspecified. This misspecification is precisely what necessitates the calibration of $\eta$; without it, the posterior uncertainty is tied to a ``wrong" likelihood, leading to potential coverage failures.

\subsubsection*{The Extended Asymmetric Laplace Model  (Ext-AL)}
To provide a more flexible Bayesian benchmark, we also consider the Extended Asymmetric Laplace (Ext-AL) model. As noted by \cite{barata:2022} and \cite{yan:2025}, a practical limitation of the standard asymmetric Laplace distribution is that its skewness and tail behavior are rigidly linked to the quantile level $\tau$. To decouple these properties, the Ext-AL distribution is constructed as a structured normal variance–mean mixture. Specifically, the error $\epsilon$ is modeled with density
\begin{align}
f(\epsilon \mid \tau, \theta) = \int_{0}^{\infty} \text{N}(\epsilon \mid \mu(w), \sigma^2(w)) , dG(w; \theta),
\end{align}
where $G(w; \theta)$ s a mixing distribution, typically from the generalized inverse Gaussian family. The model is parameterized so that the $\tau$-th quantile of $\epsilon$ is fixed at zero, while allowing flexible skewness and tail behavior. By incorporating the Ext-AL model as a comparator, we assess whether a more flexible likelihood specification alone is sufficient to deliver frequentist tolerance guarantees. However, as demonstrated in the following simulations (Section 5.2), even with this added flexibility, the Ext-AL model remains a working likelihood without an explicit calibration mechanism for frequentist coverage, leading to unstable performance under heavy-tailed or strongly asymmetric settings.

\subsubsection*{Inference and Implementation for the Bayesian Models}

For the one-sided experiments, we employ weakly informative priors to ensure that inference is primarily data-driven. For the BQR-AL model, we use the normal-scale mixture representation implemented in the \texttt{bayesQR} package \citep{benoit:2017} in \texttt{R}. For the intercept-only model, we assign a diffuse Normal prior to the quantile, $Q_\tau \sim \text{N}(0, 10^6)$. Following the approach of \cite{yang:2015}, the scale parameter is fixed at $\sigma = 1$, effectively treating the check loss as a fixed working likelihood kernel for the initial posterior estimation. To address the inherent misspecification of the asymmetric Laplace likelihood, we apply a normal approximation to the resulting MCMC output, using the posterior mean together with an asymptotically valid sandwich-type covariance matrix to construct the final $(P, 1-\alpha)$ tolerance limits.

The Ext-AL model is implemented using the hierarchical prior specification recommended by \cite{yan:2025}. The quantile $Q_\tau$ is assigned a Normal prior centered at the empirical $\tau$-quantile of the sample with variance 10. The scale parameter $\sigma$ is given an $\text{Inverse-Gamma}(2, 2)$ prior, which yields a finite mean but infinite variance, allowing for substantial flexibility in scale estimation. The shape parameter $\gamma$ is assigned a Uniform prior, $\text{U}(u_{\gamma}, v_{\gamma})$, with bounds numerically determined for each $\tau$ to ensure that the $\tau$-th quantile remains fixed at zero.

Posterior inference for both models is conducted via Markov Chain Monte Carlo. Once posterior samples are obtained, the $(P, 1-\alpha)$ one-sided tolerance limits are constructed based on the posterior distribution of the relevant quantile. Specifically, the upper tolerance limit is defined as the $(1-\alpha)$ quantile of the posterior distribution of $Q_P$, while the lower tolerance limit is defined as the $\alpha$ quantile of the posterior distribution of $Q_{1-P}$. Because these Bayesian benchmarks exhibit poor performance in the one-sided setting, particularly failing to provide stable coverage in the heavy-tailed Pareto and heavy-tailed mixture scenarios (see Section 5.2), they are omitted from the two-sided evaluations (Section 5.3), where similar behavior is expected.

\end{document}